\documentclass{article}
\usepackage[affil-it]{authblk}

\usepackage{makeidx}

\usepackage{graphicx, mathtools, amssymb, mathrsfs}
\usepackage[margin=1.0in]{geometry}
\usepackage{url}
\usepackage{amsmath}
\usepackage{graphicx}
\usepackage{tikz-cd}

\usepackage[makeroom]{cancel}

\newtheorem{theorem}{Theorem}[section]

\newtheorem{corollary}[theorem]{Corollary}

\newenvironment{proof}[1][Proof]{\begin{trivlist}
\item[\hskip \labelsep {\bfseries #1}]}{\end{trivlist}}
\newenvironment{definition}[1][Definition]{\begin{trivlist}
\item[\hskip \labelsep {\bfseries #1}]}{\end{trivlist}}

\newcommand{\qed}{\nobreak \ifvmode \relax \else
      \ifdim\lastskip<1.5em \hskip-\lastskip
      \hskip1.5em plus0em minus0.5em \fi \nobreak
      \vrule height0.75em width0.5em depth0.25em\fi}

\let\oldsqrt\sqrt
\def\sqrt{\mathpalette\DHLhksqrt}
\def\DHLhksqrt#1#2{%
\setbox0=\hbox{$#1\oldsqrt{#2\,}$}\dimen0=\ht0
\advance\dimen0-0.2\ht0
\setbox2=\hbox{\vrule height\ht0 depth -\dimen0}%
{\box0\lower0.4pt\box2}}

\newcommand{\la}{\langle}
\newcommand{\ra}{\rangle}

\makeindex

\begin{document}
%\frontmatter  

%\title{Techniques For Proving Quantum Speed Limit Formulas Including the Margolus-Levitin Theorem and Generalisations}
\title{A Geometrical Derivation of \\a Family of Quantum Speed Limit Results}
\author{Benjamin Russell, Susan Stepney}
\affil{Department of Computer Science, University of York, UK, Y010 5DD}
\date{\today}

\maketitle

\begin{abstract}
We derive a family of quantum speed limit results in time independent systems with pure states and a finite dimensional state space,
by using a geometric method based on right invariant action functionals on $SU(N)$.
The method relates speed limits for implementing quantum gates to bounds on orthogonality times.
We reproduce the known result of the Margolus Levitin theorem, and a known generalisation of the Margolis-Levitin theorem, as special cases of our method, which produces a rich family of other similar speed limit formulas corresponding to positive homogeneous functions on $\mathfrak{su}(n)$.
We discuss the general relationship between speed limits for controlling a quantum state and a system's time evolution operator. 
\end{abstract}

%%%%%%%%%%%%%%%%%%%%%%%%%%%%%%%%%%%%%%%%%%%%%%%%%%%%%%%%%%%%%%%%%%%%%%%%%%%%%%%%%
\section{Problem and Motivation}

There is much interest in the speed limit to QIP tasks; a range of perspectives can be found in \cite{ACAR, NFINQSL, MLB, OCQSL, QCOCU}.
There is also recent interest in applications of geometry to such issues \cite{Brody, GSL, CGSL}.

The first of the central problems concerning the quantum speed limit (QSL) is how quickly the time evolution operator of a given system can be driven to a certain desired operator $\hat{O}$.
This problem is the topic of many papers, notably \cite{ACAR, NFINQSL, me1, me2, Meier, bro2}, which apply techniques from differential geometry.

The second central problem is that of controlling a system's state $|\psi_0 \ra$ so that it reaches a desired state $|\psi_1\ra$ in the fastest possible time.
This question has a special instance for orthogonal states, $\la \psi_1 |\psi_0\ra =0$.
This case has attracted much attention since shortly after the dawn of quantum theory \cite{MLB, PB}.
More recently it has gained new attention due to its relevance to computer science and the physical limits to computation \cite{Lloyd, mar}.
Lloyd \cite{Lloyd} provides a short discussion of the importance of orthogonal states (in the standard inner product of $\mathbb{C}^N$) to assessing the capacity of a physical system for QIP:
there are pairs of (pure) states which are distinguishable with a single quantum measurement, at least in principle.

Additional interesting work on both these problems, hinting at the mathematical connection between them, can be found in \cite{ACAR}.
Here we focus on the relationship between the two problems in time independent systems, and on showing that well known bounds like the Margolus-Levitin theorem \cite{MLB} can be thought of as speed limits on $SU(N)$.
This view is clarifying and unifies the two seemingly separate issues:
that of controlling $\hat{U}_t$ and of orthogonalising a given state.

The structure of this paper is as follows.
First we review existing QSL results (\S\ref{sec:orthog}).
In \S\ref{sec:sun} we derive a speed limit theorem for positive homogeneous functions on $\mathfrak{su}(n)$.
Then we apply this theorem to derive a family of limits for orthogonality times in \S\ref{sec:deriv}, and show that the ML bound and Time-Energy uncertainty relation are special cases of this family.
In \S\ref{sec:connect} we discuss the connection between such speed limits for the two central problems.

%%%%%%%%%%%%%%%%%%%%%%%%%%%%%%%%%%%%%%%%%%%%%%%%%%%%%%%%%%%%%%%%%%%%%%%%%%%%%%%%%
\section{Review of Orthogonality Times}
\label{sec:orthog}

We briefly review two well known bounds on orthogonality times, the Margolus-Levitin theorem and the Time-Energy uncertainty relation.

\subsection{The Margolus-Levitin Theorem}
The Margolus-Levitin theorem \cite{MLB} (ML theorem) states that for any quantum system such that
\begin{enumerate}
 \item the Hamiltonian $\hat{H}$ is non-degenerate and has a discrete spectrum
 \item the state $|\psi \ra$ has energy expectation $\overline{E} = \la \psi | \hat{H} | \psi \ra$
\end{enumerate}
then $t_{\bot}$, the fastest possible time in which the system can transition from $|\psi \ra$ to a state orthogonal to $|\psi \ra$, is bounded by:
\begin{align}
 t_{\bot} \geq \frac{\pi \hbar}{2 (\overline{E} - E_0)}
\end{align}
where $E_0$ is the lowest eigenvalue of $\hat{H}$.
$t_{\bot}$ is also referred to as a passage time \cite{Brody}.

A fully rigorous presentation on the ML theorem requires a technical discussion about the Hamiltonian being densely defined and bounded below.
The exact nature of the spectrum also needs to be restricted since the term ``eigenvalue'' may not  make sense in the case of continuous spectra (e.g. a single free particle).
However, all the work in this paper is exclusively concerned with finite dimensional systems, and thus such technicalities are avoided.
We intend to pursue future work in that direction in order to assess the computational capacities of continuous variables systems and to compare them to finite dimensional ones using geometric methods.

\subsection{Time-Energy Uncertainty}

The Time-Energy `Uncertainty Relation' states:
\begin{align}
 t_{\bot} \geq \frac{\pi \hbar}{2 \Delta E}
\end{align}
where $\Delta E$ is the variance of energy, or `energy uncertainty', given by
\begin{align}
\Delta E = \sqrt{\la \psi |\hat{H}^2 | \psi \ra - \la \psi |\hat{H} |\psi \ra^2}.
\end{align}

A detailed discussion of this relation can be found in \cite{PB}.
Here it suffices to say that the interpretation in terms of orthogonal states is similar to the interpretation of the ML theorem.

%%%%%%%%%%%%%%%%%%%%%%%%%%%%%%%%%%%%%%%%%%%%%%%%%%%%%%%%%%%%%%%%%%%%%%%%%%%%%%%%%
\section{Homogeneous Functions on $\mathfrak{su}(n)$ and the QSL}
\label{sec:sun}
In this section we obtain a speed limit formula for each positive homogeneous function (PH function) on $\mathfrak{su}(n)$.
First we need some definitions.

\subsection{Definitions: Homogeneous Functions}

\begin{definition}[Positive Homogenous Function]
Given a vector space $V$ over $\mathbb{R}$, a function $F:V \rightarrow \mathbb{R}$ is called \emph{positive homogeneous} if $\forall v \in V, \forall \lambda \in \mathbb{R}, \lambda > 0$:
\begin{align}
 F(\lambda v) = \lambda F(v)
\end{align}
\end{definition}
Note that
a positive homogeneous function is not necessarily a positive {\it function};
the positivity refers to $\lambda$, not $F$.

\begin{definition}[Absolutely Homogenous Function]
Given a vector space $V$ over $\mathbb{R}$, a function $F:V \rightarrow \mathbb{R}$ is called \emph{absolutely homogeneous} if $\forall v \in V, \forall \lambda \in \mathbb{R}$:
\begin{align}
 F(\lambda v) = |\lambda| F(v)
\end{align}
\end{definition}
All norms are absolutely homogeneous by definition.
Being absolutely homogeneous is a stronger requirement: 
all absolutely homogeneous functions are also positive homogeneous.
So the results derived below for PH functions apply for norms.

\subsection{Speed Limits}

\subsubsection{Right invariant action functional on $SU(N)$}

Given a PH function on $\mathfrak{su}(n)$ one can define a right invariant action functional on $SU(N)$ by right extension \cite{bump}.
Let $F$ be a PH function on $\mathfrak{su}(N)$ and define its right extension $F_{\hat{U}}: T_{\hat{U}}SU(N) \rightarrow \mathbb{R}$ by
\begin{align}
 F_{\hat{U}}(\hat{A}) := F(\hat{A} \hat{U}^{\dagger})
\end{align}
This is also known as the canonical lift of right translation from $SU(n)$ to $TSU(n)$.

Every tangent vector $\hat{A} \in T_{\hat{U}}SU(N)$ can be expressed as $\hat{A} = \hat{B} \hat{U}$ for some $\hat{B} \in \mathfrak{su}(N)$.
This is essentially the existence of a right trivialisation of the tangent bundle for a matrix Lie group, $TSU(n) \cong \mathfrak{su}(n) \times SU(n)$, which exists as all Lie groups are parallelisable \cite[\S 5.1]{LiePar}.

Exploiting the right trivialisation proves helpful later, as does the following observation which holds for right invariant $F$:
\begin{align}
 F_{\hat{U}}(\hat{A}\hat{U}) = F(\hat{A} \hat{U}^{\dagger} \hat{U}) = F(\hat{A})
\end{align}
That is, all functions that are the right extension of a function on $\mathfrak{su}(N)$ are right invariant.
In fact, right extension can yield all such right invariant functions on $TSU(N)$, which can be readily confirmed.

Given any right invariant PH function on $TSU(N)$, it is possible to define an action functional $S$ for curves $\hat{U}_t$ on $SU(N)$:
\begin{align}
 S[\hat{U}_t] = \int_{0}^{T} F_{\hat{U}_t} \left(\frac{d \hat{U}_t}{dt} \right) dt
\end{align}
This action functional is itself right invariant.

Consider the case where $\hat{U}_t$ solves the Schr{\"o}dinger equation
\begin{align}\label{eqn:schro}
 \frac{d \hat{U}_t}{dt} = -i\hat{H}_t \hat{U}_t
\end{align}
for a potentially time dependent Hamiltonian $\hat{H}_t$.
Then a right invariant action can be expressed as:
\begin{align}
 S[\hat{U}_t] & = \int_{0}^{T} F_{\hat{U}_t} \left(\frac{d \hat{U}_t}{dt} \right) dt = \int_{0}^{T} F_{\hat{U}_t} \left(-i\hat{H}_t \hat{U}_t \right) dt = \int_{0}^{T} F \left(-i\hat{H}_t \hat{U}_t \hat{U}_t^{\dagger} \right) dt = \int_{0}^{T} F \left(-i\hat{H}_t \right) dt
\end{align}
In the case that $\hat{H}_t$ is not time dependent, that is $\hat{H}_t = \hat{H}$ for all time, this reduces to:
\begin{align}
 S[\hat{U}_t] & = \int_{0}^{T} F \left(-i\hat{H}_t \right) dt \nonumber = \int_{0}^{T} F \left(-i\hat{H} \right) dt = T F \left(-i\hat{H} \right) dt
\end{align}

Riemannian metrics on a manifold $M$ are PH functions (on $T_pM$ for each $p$), as are Finsler metrics; for relevant applications of such structures see \cite{NFINQSL, mese}.
PH functions have a favorable property that is exploited throughout their use in geometry:
the action corresponding to a PH (point-wise on each tangent space of a manifold) function does not depend on the parametrisation of that curve.
This in Riemannian/Finsler geometry is the statement that arc-length does not depend on the parametrisation of a curve.
This is an easily checked for any PH function, and the proof is identical to the corresponding proof typically given in Riemannian geometry \cite{pet}, so it is not reproduced here.
This is the motivation for considering PH functions.

For functions $F$ such that $F(-v) \neq F(v)$, such as non-reversible Finsler metrics, care must be taken regarding exactly which re-parametrisation are allowed.
The action of a curve on $SU(N)$ according to the right translation of such a function on $\mathfrak{su}(N)$ is invariant under re-parametrisations that are ``positive'':
that is, if $r$ is the original parameter, and $s$ is the new one, then ${ds}/{dr} > 0$ is required at all points along the curve.
A key example of non-reversible Finsler metrics are the Randers metrics \cite{ran}.

\subsubsection{One Parameter Sub-Groups of $SU(N)$ are the Time Independent Trajectories}

The one parameter subgroups of $SU(N)$ are known by a trivial (as we are in finite dimension) application of Stone's theorem \cite{stone}.
That theorem implies that one parameter subgroups of $SU(N)$ are all of the form:
\begin{align}\label{eqn:opsubgp}
 \hat{U}_t = e^{-it\hat{K}}
\end{align}
for some Hermitian matrix $\hat{K}$.
One can readily check that such a $\hat{U}_t$ satisfies the group axioms for a one parameter group.
Most importantly, $\hat{U}_s \hat{U}_t = \hat{U}_{s+t}$ holds, as does $\hat{U}_0 = \hat{I}$, which in part facilitates the interpretation of the parameter as physical time.

It is also readily checked that a $\hat{U}_t$ defined by eqn.(\ref{eqn:opsubgp}), solves the Shr\"odinger equation for a time independent Hamiltonian.
That is, given a Hamiltonian $\hat{H}$, then $\hat{U}_t = e^{-it\hat{H}}$ is the time evolution operator, since this $\hat{U_t}$ solves the Shr\"odinger equation (eqn(\ref{eqn:schro})).
%\begin{align}
%\frac{d}{dt} \hat{U}_t = -i\hat{H}\hat{U}_t 
%\end{align}

\subsubsection{Action of A Time Independent Trajectory With A Given Endpoint}

Suppose a gate $O$ is implemented in a quantum system with time independent Hamiltonian $\hat{K}$, and that the gate takes time $T$ to implement.
That is:
\begin{align}
 \hat{U}_T = e^{-iT \hat{K}} =\hat{O}
\end{align}
This implies (by taking logs and rearranging) that the relevant Hamiltonian is:
\begin{align}
  \hat{K} = \frac{i}{T} \log(\hat{O})
\end{align}
%It is clear that setting the Hamiltonian $\hat{K} = \frac{i}{T} \log(\hat{O})$ does implement $\hat{O}$ at time $t=T$.
%This is because $\hat{U}_t = e^{\frac{t}{T} \log(\hat{O})}$ which at time $t=T$ comes to $e^{\log(\hat{O})} = \hat{O}$.

We can now find the action of the curve $\hat{U}_t$ connecting $\hat{I}$ to $\hat{O}$ along a time independent trajectory as follows.
Let $F$ be a right invariant PH function on $T\,SU(N)$.
Then
\begin{align}
 S[\hat{U}_t] & = \int_{0}^{T} F_{e^{-it \hat{K}}}\left(\frac{d}{dt} e^{-it \hat{K}}\right) dt = \int_{0}^{T} F\left(-i\hat{K} \right) dt = T F\left(-i \frac{i}{T} \log(\hat{O}) \right) = F\left(\log(\hat{O}) \right) \nonumber
\end{align}
In the final step the $T$s cancel due to the assumed homogeneity of $F$.
As any such action is invariant under reparameterisation of the one parameter subgroup connecting $\hat{I}$ to $\hat{O}$, we have the following theorem:
\begin{theorem}
\label{thm1}
Given any PH function $F: \mathfrak{su}(N)$, then any time independent, finite dimensional quantum system with Hamiltonian $\hat{H}$ such that $\hat{U}_T = \hat{O}$ satisfies:
\begin{align}
 T = \frac{F\left(\log(\hat{O}) \right)}{F \left(-i\hat{H} \right)}
\end{align} 
\end{theorem}
\begin{proof}
Any two parameterisations of any curve must yield the same value for the action, as $F$ is a PH function.
$T F \left(-i\hat{H} \right)$ and $F\left(\log(\hat{O}) \right)$ are two different formulas for the action for two paramerterisations of the same curve, hence they must be equal.
$\Box$
\end{proof}

It should be noted that an implicit assumption about $F$ has been made.
It is assumed that $F$ is non-singular at $\log(\hat{O})$ and non-zero at $\log(-i\hat{H})$ so that the expression for $T$ is finite.

Theorem (\ref{thm1}) has a corollary.
\begin{corollary}
If the Hamiltonian is constrained such that $F(-i \hat{H}) = \kappa$ for some $\kappa \in \mathbb{R}$ then:
\begin{align}
 T_{\text{opt}} = \frac{1}{\kappa} F\left(\log(\hat{O}) \right)
\end{align}
\end{corollary}

\subsection{Example PH functions}

\subsubsection{Constraining the $p$th Central Moment of an Observable}

Given a fixed, normalised state $|\psi \ra$ in $\mathbb{C}^N$ one can now ask the following question:
How quickly can a time independent quantum system implement a given gate $\hat{O}$ given that it is constrained to energy expectation $\overline{E} - E_0 =\kappa$ for some given, fixed $\kappa$?

Let $G^{(|\psi \ra)}:\mathfrak{su}(N) \rightarrow \mathbb{R}$ be defined by:
\begin{align}
 G^{(|\psi \ra)}(-i\hat{H}) = \frac{\la \psi | \hat{H} - E_0 \hat{I} | \psi \ra}{\la \psi | \psi \ra} = \overline{E} - E_0
\end{align}
where $E_0$ is the lowest eigenvalue of $\hat{H}$.
This function is a special case ($p=1$) of a more general $G_p, p>0$:
\begin{align}
 G^{(|\psi \ra)}_{p}(-i\hat{H}) = \frac{\left(\la \psi | (\hat{H} - E_0 \hat{I})^p | \psi \ra\right)^{1/p}}{\la \psi | \psi \ra}
\end{align}
%$G^{(|\psi \ra)}_{p}(\lambda (-i\hat{H}))$ is the energy expectation ($p=1$) and the energy uncertainty ($p=2$) 
$G_{1}$ is the energy expectation and $G_{2}$ is the energy uncertainty.
We have that $\forall p>0$, $\forall \lambda>0$:
\begin{align}
 G^{(|\psi \ra)}_{p}(\lambda (-i\hat{H})) & = \frac{(\la \psi | (\lambda \hat{H} - \lambda E_0 \hat{I})^p | \psi \ra)^{1/p}}{\la \psi | \psi \ra} = \frac{(\lambda^p \la \psi | (\hat{H} - E_0 \hat{I})^p | \psi \ra)^{1/p}}{\la \psi | \psi \ra} = \lambda G^{(|\psi \ra)}_{p}(-i\hat{H}) \nonumber
\end{align}
Hence all the $G_p$ are PH functions on $\mathfrak{su}(N)$ for a fixed state and value of $p$,
so we can apply theorem (\ref{thm1}). 

\subsubsection{Speed Limit}

We now apply theorem (\ref{thm1}) to a gate $\hat{O}$ and a PH function $G^{(|\psi \ra)}_{p}$.
If a time independent system is constrained such that $G^{(|\psi \ra)}_{p}(-i\hat{H}) = \kappa$ then $T$, the time to implement gate $\hat{O}$, is:
\begin{align}
 T = \frac{1}{\kappa} G^{(|\psi \ra)}_{p}(\log(\hat{O}))
\end{align}

%%%%%%%%%%%%%%%%%%%%%%%%%%%%%%%%%%%%%%%%%%%%%%%%%%%%%%%%%%%%%%%%%%%%%%%%%%%%%%%%%
\section{Deriving Orthogonality Times}
\label{sec:deriv}

We apply theorem (\ref{thm1}) to derive specific orthogonality times for specific systems.

\subsection{Two Level ML Bound}
\label{sec:2ML}
We provide the following simple example in order to ease into the general case, and to illustrate the method.
In a two level system we can prove a bound on an orthogonality time using theorem (\ref{thm1}) by setting:
\begin{align}
\label{swap}
\hat{O} = 
e^{i \pi/2} \begin{pmatrix}
0 & e^{-i \theta} \\
e^{i \theta} & 0
\end{pmatrix}
\end{align}
Diagonalising gives a matrix of the form:
\begin{align}
\hat{O} = \hat{S} \hat{J} \hat{S}^{-1}
\end{align}
where:
\begin{align}
\hat{S} = \begin{pmatrix}
-e^{i \theta} & e^{i \theta} \\
1 & 1
\end{pmatrix}
\quad \quad
\hat{J} = \begin{pmatrix}
-i & 0 \\
0  & i
\end{pmatrix}
\quad \quad
\hat{S}^{-1} = \frac{1}{2} \begin{pmatrix}
-e^{-i \theta} & 1 \\
e^{-i \theta} & 1
\end{pmatrix}
\end{align}
Applying the fact that the matrix logarithm is analytic one finds:
\begin{align}
\log(\hat{O}) = \hat{S} \log(\hat{J}) \hat{S}^{-1} = \hat{S} \begin{pmatrix}
\log(-i) & 0 \\
0  & \log(i)
\end{pmatrix} \hat{S}^{-1}
= \frac{\pi}{2} \hat{S} \begin{pmatrix}
-i & 0 \\
0 & i
\end{pmatrix} \hat{S}^{-1} = \frac{\pi}{2} \hat{O}
\end{align}

In order to answer how quickly a time independent system such that $G^{(|\psi \ra)}_{1}(-i\hat{H}) = \kappa$ (here choosing $p=1$ as an example) can implement $\hat{O}$, we evaluate:
\begin{align}
 T & = \frac{1}{\la \psi | \psi \ra} \frac{1}{\kappa}G^{(|\psi \ra)}_{p}(\log(\hat{O})) \\
   & = \frac{1}{\la \psi | \psi \ra} \frac{1}{\kappa}G^{(|\psi \ra)}_{p}(\hat{S} \log(\hat{J}) \hat{S}^{-1}) \nonumber \\
   & = \frac{1}{\la \psi | \psi \ra} \frac{\pi}{2 \kappa} \la \psi | \hat{S} \begin{pmatrix} 1 & 0 \\ 0 & -1 \end{pmatrix} \hat{S}^{-1} + \hat{I} | \psi \ra \nonumber \\
   & = \frac{\pi}{2 \kappa} + \frac{1}{\la \psi | \psi \ra} \frac{\pi}{2 \kappa} \la \psi | \hat{S} \begin{pmatrix} 1 & 0 \\ 0 & -1 \end{pmatrix} \hat{S}^{-1}| \psi \ra \nonumber \\
   & = \frac{\pi}{2 \kappa} + \frac{1}{\la \psi | \psi \ra} \frac{i}{\kappa} \la \psi | \hat{O} | \psi \ra \nonumber \\
\end{align}

This always results in a real valued time, because $\hat{O}$ is anti-Hermitian and thus has  purely imaginary expectation in any state.
We now observe that $\la \psi | \hat{O} | \psi \ra = 0$ is equivalent to saying that $\hat{O}$ maps $| \psi \ra$ to an orthogonal state.
$T$ is the optimal time (by theorem \ref{thm1}) for this to happen, and $T_{{\bot}}$ is thus given by:
\begin{align}
\label{eqn1}
 T_{\bot} = \frac{\pi}{2 \kappa}
\end{align}
$\kappa$ evaluates to:
\begin{align}
G^{(|\psi \ra)}_{1}(-i\hat{H})) = \kappa
\end{align}
which implies that
\begin{align}
 \kappa = \frac{\la \psi | \hat{H} - E_0 \hat{I} | \psi \ra}{\la \psi | \psi \ra} = \overline{E} - E_0
\end{align}
and thus eqn.(\ref{eqn1}) becomes
\begin{align}
 T = \frac{\pi}{2 (\overline{E} - E_0)}
\end{align}
which is the ML bound for this system, but here derived using our method which is significantly different to the original method.

\subsection{A Bound in Terms of the $p$'th Fractional Moment}

{Zieli{\'n}ski} \& {Zych} \cite{Zych} provide a bound in terms of the $p$'th fractional moment of energy.
We reproduce that bound here using our different, geometrical proof technique.

A theorem is needed in this section:
\begin{theorem}
\label{thm2}
 Given any special unitary operator $\hat{T}$ on $\mathbb{C}^N$ such that the following holds:
\begin{itemize}
 \item $\exists |\psi_0 \ra, |\psi_1 \ra \in \mathbb{C}^N$ such that $\hat{T} |\psi_0 \ra = |\psi_1 \ra$ and $\hat{T} |\psi_1 \ra = |\psi_0 \ra$
 \item $\la \psi_1 |\psi_0 \ra = 0$
 \item $\forall |\psi \ra \in \{ | \psi_0 \ra, |\psi_1 \ra \}^{\bot}$, $\hat{T} |\psi \ra = |\psi \ra$  (where $\bot$ indicates the orthogonal complement)
\end{itemize}
Then $\hat{T}$ has the form
\begin{align}
 \hat{T} = \hat{V} \hat{O} \hat{V}^{\dagger}
\end{align}
where
\begin{align}
\hat{O} = \left( \begin{pmatrix}
0 & i e^{-i \theta} \\
i e^{i \theta} & 0
\end{pmatrix} \oplus \hat{I}_{N-2} \right)
\end{align}
for some special unitary $\hat{V}$.
\end{theorem}
That is, given any $\hat{T}$ that sends a specific pair of orthogonal states to each other, and leaves all other orthogonal states unchanged, then $\hat{T}$ can be expressed in a simplified form,
of some unitary conjugation of the specific $\hat{O}$ given in the theorem.

\begin{proof}[Proof Sketch]
Choose $\hat{V}$ to be a change of basis matrix (all unitary matrices are change of basis matrices between orthonormal bases of $\mathbb{C}^N$).
More specifically, take $\hat{V}$ to be a change from an orthonormal basis that includes $|\psi_0 \ra$ and $|\psi_1 \ra$ to a basis that includes $|0\ra$ and $|1\ra$.
Choose a $\hat{V}$ that also satisfies $\hat{V} |\psi_0 \ra = |0\ra$ and $\hat{V} |\psi_1\ra = |1\ra$.
The theorem follows from this choice. $\Box$
\end{proof}

Now we can derive a bound in terms of the $p$'th ($p > 0$) fractional moment of energy that applies to systems with arbitrary (finite) number of levels.
This is achieved by applying theorems (\ref{thm1}) and (\ref{thm2}) in full generality.

In what follows $|0 \ra$ is taken to mean the normalised state:
\begin{align}
|0 \ra := \begin{pmatrix}
1 \\ 0 \\ \vdots \\ 0
\end{pmatrix}
\end{align}
As $SU(N)$ acts transitively on the normalised states, $\hat{V} \in SU(N)$ can be chosen so that $\hat{V}|0\ra$ represents an arbitrary (normalised) state.
Any other state could play the role of $|0\ra$ here; it is simply a convenient choice.

Consider the constraint on $\hat{H}$, that $G^{\hat{V} | 0 \ra}_{p} (-i\hat{H}) = \kappa$, and the optimal implementation time for the gate $\hat{O}_{\hat{V}}=\hat{V}\hat{O}\hat{V}^{\dagger}$ as it appears in theorem (\ref{thm2}).
Calculate $\log(\hat{O}_{\hat{V}})$:
\begin{align}
\label{form1}
 \log(\hat{O}_{\hat{V}}) = \log(\hat{V}\hat{O}\hat{V}^{\dagger}) = \hat{V}\log(\hat{O})\hat{V}^{\dagger} = \frac{\pi}{2} \hat{V} \left( \begin{pmatrix} 0 & i e^{-i \theta} \\ i e^{i \theta} & 0 \end{pmatrix} \oplus \hat{Z} \right) \hat{V}^{\dagger}
\end{align}
where $\hat{Z}$ is the zero matrix of the appropriate size.
By theorem (\ref{thm2}) this covers all operators that swap exactly a pair of states.

The optimal time to implement this in a system such that $G^{\hat{V} | 0 \ra}_{p} (-i\hat{H}) = \kappa$ is, by theorem (\ref{thm1}):
\begin{align}
 T & = \frac{1}{\kappa} G^{\hat{V} | 0 \ra}_{p}(\log(\hat{O}_{\hat{V}})) \\ 
   & = \frac{1}{\kappa} \left( \la 0 | \hat{V}^{\dagger} \left( \frac{i \pi}{2} \hat{V} \left( \begin{pmatrix} 0 & i e^{-i \theta} \\ i e^{i \theta} & 0 \end{pmatrix} \oplus \hat{Z} \right) \hat{V}^{\dagger} + \frac{\pi}{2} \hat{I} \right)^{p} \hat{V} |0\ra \right)^{1/p} \nonumber \\
   & = \frac{\pi}{2 \kappa} \left( \la 0 | \left( \begin{pmatrix} 0 & -e^{-i \theta} \\ -e^{i \theta} & 0 \end{pmatrix} \oplus \hat{Z} + \hat{I} \right)^{p} |0\ra \right)^{1/p} \nonumber \\
   & = \frac{\pi}{2 \kappa} \left( \la 0 | \begin{pmatrix} 1 & -e^{-i \theta} \\ -e^{i \theta} & 1 \end{pmatrix}^{p} \oplus \hat{I} |0\ra \right)^{1/p} \nonumber \\
   & = \frac{\pi}{2 \kappa} \left( (1 0 \cdots 0) \begin{pmatrix} 2^{p-1} & -2^{p-1} e^{-i \theta} \\  -2^{p-1}e^{i \theta} & 2^{p-1} \end{pmatrix} \oplus \hat{I} \begin{pmatrix} 1 \\ 0 \\ \vdots \\ 0 \end{pmatrix} \right)^{1/p} \nonumber \\
   & = \frac{\pi}{ \kappa\, 2^{{1}/{p}}} \nonumber
\end{align}

This coincides with the ML bound exactly for $p=1$ (as $\kappa$ is $\overline{E} - E_0$) and coincides with the generalised ML bound \cite[eqn.10]{Zych} for any $p>1$.

For $p=2$ the bound does not coincide with the time-energy uncertainty relation: there is an  extra factor in the denominator.
Better understanding this is the goal of further work.

\subsection{Two Levels with Arbitrary Initial State}

By applying theorem \ref{thm2} to a two level system we can prove the ML bound for an arbitrary state in a two level system.
In this case the operator $\hat{O}_{\hat{V}} = \hat{V}\hat{O} \hat{V}^{\dagger}$ swaps the states $\hat{V} |0 \ra$ and $\hat{V} |1 \ra$ (where $|0\ra, |1 \ra$ are the computational basis states).
As the action of the special unitary group $SU(N)$ on normalised states (that is, the sphere $\mathbb{S}_{2N-1} \subset \mathbb{C}^N$) is transitive \cite{deng}, then $\hat{V} |0\ra$ and $\hat{V} |1\ra$ in fact includes all pairs of orthogonal states.

Repeating the analysis of section~\ref{sec:2ML} {\it mutatis mutandis} yields the ML bound.
%The full derivation is not included, due to similarity with previous derivation.
One can repeat the procedure, except with the constraint
\begin{align}
 G^{\hat{V} |0\ra}_{1} (-i\hat{H}) = \kappa
\end{align}
This yields the speed limit for orthogonality $\hat{V} |0\ra$.
That bound comes out to the ML bound when applied to the gate $\hat{O}_{\hat{V}}$.

\subsection{Operator Norm}

The operator norm of a complex $N \times N$ matrix is a PH function, as it is a norm \cite{horn}.
It is in fact an absolutely homogeneous function.
There is potentially a different operator norm for each norm on $\mathbb{C}^N$;
 here we consider only the operator norm corresponding to the norm arising from the standard inner product on $\mathbb{C}^N$.

Many equivalent definitions exist for the operator norm; the following is sufficient for matrices.
The operator norm $||\cdot||_{op}$ of a matrix is defined by
\begin{align}
 ||\hat{A}||_{op}^2 = \max \left\{  \frac{\la \psi| \hat{A}^{\dagger} \hat{A} | \psi \ra}{\la \psi | \psi \ra} , \forall |\psi \ra \in \mathbb{C}^N \right\}
\end{align}
This is equal to the largest singular value of $\hat{A}$, often written $\sigma_{max}(\hat{A})$.
It is unitarily invariant, that is $||\hat{V}A \hat{V}^{\dagger}||_{op} = ||\hat{A}||_{op}$ for any unitary $\hat{V}$.
This property of unitary invariance is also shared by the $G_p^{|\psi \ra}$.

We define a PH function on $\mathfrak{su}(N)$ in terms of the operator norm:
\begin{align}
 G_{op}(-i\hat{H}) = ||\hat{H} - E_0 \hat{I} ||_{op} = E_{\text{max}} - E_0
\end{align}
We calculate $G_{op}(\log(\hat{O}_{\hat{V}}))$ for $\hat{O}$ as given by eqn.(\ref{swap}) by applying eqn.(\ref{form1}) and the unitary invariance of the norm as follows:
\begin{align}
 G_{op}(\log(\hat{O}_{\hat{V}})) 
& = \left|\left|\frac{\pi}{2}(i\hat{V}\hat{O}\hat{V}^{\dagger} + \hat{I}) \right|\right|_{op} 
= \left|\left|\frac{\pi}{2}(i\hat{O} + \hat{I}) \right|\right|_{op} 
= \frac{\pi}{2}\left|\left|i\hat{O} + \hat{I} \right|\right|_{op} 
= \frac{\pi}{2} \sigma_{max}(i\hat{O} + \hat{I}) = \pi
\end{align}
where the final step, a routine eigenvalue calculation, is omitted.

By applying theorems $(\ref{thm1}, \ref{thm2})$ similarly as for the ML bound,
this implies that the optimal time to implement the gate $\hat{O}_{\hat{V}}$ is:
\begin{align}
 T = \frac{\pi}{E_{max} - E_0}
\end{align}

%%%%%%%%%%%%%%%%%%%%%%%%%%%%%%%%%%%%%%%%%%%%%%%%%%%%%%%%%%%%%%%%%%%%%%%%%%%%%%%%%%%%%%%
\section{General Connection between Speed Limits for the Two Central Problems}
\label{sec:connect}

We have used geometrical methods on operators to calculate orthogonality times on states.
There is a deep connection between the two, explored in this section.

\subsection{State Space as a Coset Space and Quantum Dynamics}

The space of physically inequivalent quantum states associated to a finite dimensional quantum system can be understood to be a complex projective space \cite{gqm}.
This space, $\mathbb{C}P^{N-1}$, can be seen as a quotient of $\mathbb{C}^{N}$ by an equivalence relation $\sim$ representing physically equivalent states.
This relation $\sim\ \subseteq \mathbb{C}^{N} \times \mathbb{C}^{N}$ is given by: 
\begin{align}
| \psi_0 \ra \sim | \psi_1 \ra  \mbox{ iff } | \psi_0 \ra = Z | \psi_1 \ra  \mbox{ for some } Z \in \mathbb{C}/\{0\}
\end{align}
That is, complex vectors that are part of the same complex line are physically the same state.
This is no more than the statement that physical states are normalised, and that the global phase of a state is not physically observable.
So one can write a physical state as $[ |\psi \ra]_{\sim}$.
In the remainder of this paper we leave the $\sim$ implicit, as no confusion is possible, and we write the state as $[ |\psi \ra]$.

$SU(N)$ has a group action on $\mathbb{C}P^{N-1}$ inherited from its group action on $\mathbb{C}^{N}$.
The standard group action $\diamond$ of $SU(N)$ on $\mathbb{C}^{N}$ is $\hat{V} \diamond | \psi \ra = \hat{V} |\psi \ra$.
 This is nothing more than matrices acting on vectors by standard matrix multiplication.
It is a Lie group action in the usual sense \cite{bump}.

The natural choice of group action $\star$ on $\mathbb{C}P^{N-1}$ is then $\hat{V} \star [|\psi\ra] := [\hat{V}|\psi\ra]$.
The sense in which this is the ``natural'' choice of action is as follows.
If we define the quotient map $\phi: \mathbb{C}^{N} \rightarrow \mathbb{C}P^{N-1}$ by $\phi(|\psi\ra) = [|\psi \ra]$ then the following diagram commutes:

\begin{center}
\begin{tikzcd}%[scale=0.5]
    {SU(N) \times \mathbb{C}^{N}} \ar{r}{\diamond} \ar{d}[swap]{\text{id} \times \phi} & \mathbb{C}^{N} \ar{d}{\phi}\\
    {SU(N) \times \mathbb{C}P^{N-1}} \ar{r}{\star} & \mathbb{C}P^{N-1}
\end{tikzcd}
\end{center}

This allows us to think of quantum time evolution as dynamics happening on $\mathbb{C}P^{N-1}$, rather than in Hilbert space, as the time evolution operator can now act on points on $\mathbb{C}P^{N-1}$ and this action is physically equivalent to the usual Hilbert space dynamics.
The group action $\star$ is transitive on $\mathbb{C}P^{N-1}$.
This follows from the observation that the group action $\diamond$ is transitive on the sphere $\mathbb{S}^{2N-1}$ (of normalised states) embedded into $\mathbb{C}^{N}$ \cite{montg}.

It is well known from the theory of homogeneous spaces (not directly related to the other sense of homogeneous used earlier) that the following relationship holds \cite{deng}:
\begin{align}
\mathbb{C}P^{N-1} \cong SU(N)/U(N-1) = \{\hat{U} U(N-1) | \hat{U} \in SU(N) \}
\end{align}
That is, the space of states can be realised as a coset space.
Here $U(N-1)$ has the specific meaning:
\begin{align}
U(N-1) = \text{stab}([|\check{\psi} \ra]_{}) = \{\hat{U} \in SU(N) \big| \hat{U} \star [|\check{\psi}\rangle] = [|\check{\psi}\rangle] \} 
\end{align}
Here one fixed arbitrary point $[|\check{\psi}\ra] \in \mathbb{C}P^{N-1}$ has been chosen; any point could have been chosen and an isomorphic construction would result throughout all that follows.
A convenient choice can be made, namely the equivalence class of the vector in $\mathbb{C}^N$ given by:
\begin{align}
 |\check{\psi}\ra = \begin{pmatrix} 1 \\ 0 \\ \vdots \\ 0 \end{pmatrix}
\end{align}
The equivalence class of this state is now given by:
\begin{align}
 [|\check{\psi}\ra] = \left\{ \begin{pmatrix} Z \\ 0 \\ \vdots \\ 0 \end{pmatrix}, Z \in \mathbb{C}/\{0\} \right\}
\end{align}

As matrices the elements of $\text{stab}([|\check{\psi} \ra]) \cong U(N-1)$ are given by
\begin{align}
\text{stab}\left([|\check{\psi}\ra \right]) =
\left\{
\left(
\begin{array}{c c}
  \det(\hat{V})^{-1} & 0 \cdots 0 \\
  0 & \raisebox{-15pt}{{\huge\mbox{{$\hat{V}$}}}} \\[-4ex]
  \vdots & \\[-0.5ex]
  0 &
\end{array}
\right)
,
\hat{V} \in U(N-1) \right\} \cong U(N-1)
\end{align}

The isomorphism is given by the map $\gamma: SU(N)/U(N-1) \rightarrow \mathbb{C}P^{N-1}$ defined by:
\begin{align}
\label{iso}
 \gamma(\hat{U} U(N-1)) := (\hat{U} U(N-1)) \star [|\check{\psi} \ra]
\end{align}
From this definition it follows that:
\begin{align}
& \gamma(\hat{U} U(N-1)) = (\hat{U} U(N-1)) \star [|\check{\psi} \ra] = \nonumber \\ 
& \hat{U} \star [U(N-1) |\check{\psi} \ra] =  \hat{U} \star [|\check{\psi} \ra] = [\hat{U} |\check{\psi} \ra]
\end{align}

$SU(N)$ also has an action $\bullet: SU(N) \times SU(N)/U(N-1) \rightarrow SU(N)/U(N-1)$ on the quotient space $SU(N)/U(N-1)$.
This action is given by:
\begin{align}
 \hat{V} \bullet (\hat{U} U(N-1)) := (\hat{V} \hat{U}) U(N-1)
\end{align}
In a similar way to $\mathbb{C}P^{N-1}$, this lets us consider quantum dynamics on $SU(N)/U(N-1)$, as the time evolution operator can now act on this space.
In order to check that this dynamics is physically equivalent to the dynamics on $\mathbb{C}^N$ (that is, the standard Schr\"odinger formalism) and the dynamics on $\mathbb{C}P^{N-1}$, we must check that the following diagram commutes:
\begin{center}
\begin{tikzcd}%[scale=0.5]
    {SU(N) \times SU(N)/U(N-1)} \ar{r}{\bullet} \ar{d}[swap]{\text{id} \times \gamma} & SU(N)/U(N-1) \ar{d}{\gamma}\\
    {SU(N) \times \mathbb{C}P^{N-1}} \ar{r}{\star} &  \mathbb{C}P^{N-1}
\end{tikzcd}
\end{center}
This is checked by confirming that the following holds:
\begin{align}
 \star((\text{id} \times \gamma)(\hat{U}, \hat{V}U(N-1))) = \gamma(\bullet(\hat{U}, \hat{V}U(N-1)))
\end{align}
This follows directly from the above definitions of the maps involved, firstly:
\begin{align}
 \star((\text{id} \times \gamma)(\hat{U}, \hat{V}U(N-1))) = \star( (\hat{U}, [\hat{V} |\check{\psi} \ra]) ) = [\hat{U}\hat{V} |\check{\psi} \ra]
\end{align}
and secondly:
\begin{align}
\gamma(\bullet(\hat{U}, \hat{V}U(N-1))) = \gamma((\hat{U}\hat{V}) U(N-1)) = [\hat{U} \hat{V} |\check{\psi} \ra]
\end{align}
and thus the diagram commutes as the two are equal.

This allows the formulation of quantum dynamics in at least three physically equivalent ways.
One can take $\mathbb{C}^{N}$ (which includes some physical redundancy), $\mathbb{C}P^{N-1}$, or $SU(N)/U(N-1)$ as the space of states.
The standard time evolution operator (considered to be an element of $SU(N)$ for all time) can act on all of these spaces, and that these actions are equivalent.
The third of these formulations is the one that most easily allows us to connect speed limits on traveling between a pair of states with speed limits for implementing a quantum gate.

\subsection{Role of the Pushforward}
In order for a PH function $F$ on $TSU(N)$ (that is, PH on each tangent space individually) to push forward unambiguously to a PH function on $T\mathbb{C}P^{N-1}$ it is required to be `constant on cosets': $F_{\hat{U}}(\hat{A}U(N-1)) = F(\hat{A})$ $\forall \hat{U} \in SU(N), \forall \hat{A} \in \mathfrak{su}(N)$.
This is sometimes referred to as being `compatible with the quotient'.
If we also require that this pushforward of $F$ is invariant under the action of $SU(N)$, then we also require:
\begin{align}
\label{cond}
 F(\hat{V} \hat{A} \hat{U} \,U(N-1)) = F(\hat{A})
\end{align}
for any $\hat{U},\hat{V} \in SU(N)$, $\forall \hat{A} \in \mathfrak{su}(n)$.
This condition is only slightly weaker than bi-invariance.
The authors know of no PH functions on $TSU(N)$ that satisfy eqn.(\ref{cond}) but are not bi-invariant.

This sheds light on why theorem \ref{thm1} provides a speed limit for a process happening on $\mathbb{C}P^{N-1}$ corresponding to each right invariant PH function on $TSU(N)$.
These functions are all constant on cosets and thus pushforward to PH functions on $\mathbb{C}P^{N-1}$ unambiguously.
This allows them to be thought of as measuring the action of a curve on the space of states $\mathbb{C}P^{N-1}$, even though they are defined on the group $SU(N)$.
This is why they correspond to speed limits for state transfer problems generally, and thus to bounds on orthogonality times.

%%%%%%%%%%%%%%%%%%%%%%%%%%%%%%%%%%%%%%%%%%%%%%%%%%%%%%%%%%%%%%%%%%%%%%%%%%%%%%%%%%%%%%%%%%%%%%%%%%%%%%%%%%%%%%%%%%%%%%%%%%%%

\section{Conclusions and Further Work}

We have shown a novel method for deriving bounds on orthogonality times in time independent quantum systems with finite dimensional state spaces.
We have derived a general expression for such times (theorem~\ref{thm1}),
and have use it to re-derive existing results in a unifying manner.
The method also sheds light on the mathematical structures corresponding to speed limit formulas.

We believe that our method may be extended to time dependent systems by following the analysis of \cite{Meier, mese}.
This should allow us to more fully understand the relationship between the two notions of the quantum speed limit described in this work.

We feel that, in light of the observations on time optimal state control in \cite{Meier}, that a re-examination of the relationship between the work in \cite{qbr} and \cite{ACAR} would be fruitful.
We conjecture that it will be possible to extend the understanding of the connection between time optimal gates and time optimal state transfer illustrated here and in \cite{Meier}, and to show that the Lagrangian in \cite{ACAR} pushes forward to the one in \cite{qbr} with full generality.

\subsection*{Acknowledgments}
We thank Dorje Brody and David Meier of Brunel University for an interesting ongoing discussion of the geometry of the quantum speed limit and their work on this matter.
Russell is supported by an EPSRC studentship.
\bibliographystyle{plain}
\bibliography{p5}

\end{document}